\documentclass[11pt]{article}
\usepackage{amsmath,amssymb,amscd,latexsym,amsthm,mathrsfs}
\usepackage[unicode]{hyperref}
\usepackage{hypbmsec,longtable}
\usepackage{graphicx}
\usepackage{fancyhdr}
\ifx\pdfoutput\undefined\else
\hypersetup{
colorlinks=true,
linkcolor=blue,
citecolor=blue,
urlcolor=blue,
filecolor=blue,
bookmarksnumbered=true,
pdfstartview=FitH,
pdfhighlight=/N
}
\fi

\newtheorem{theorem}{Theorem}

\theoremstyle{remark}

\def\m{{\operatorname{m}}}

\usepackage{cite}

\def\m{{\rm m}}
\renewcommand{\author}[1]{\large\rm #1\\ \bigskip}
\newcommand{\address}[1]{{\normalsize\it #1\\}\bigskip}
\renewcommand{\title}[1]{\bigskip\bigskip\Large\bf #1\bigskip\bigskip\\}

\begin{document}

%\begin{flushright}
%{\sf version 1.1 \it (11/07/2012)}
%\end{flushright}

\vglue .3 cm

\begin{center}

\title{An integral arising from the chiral $sl(n)$ Potts model}
\author{              Anthony J. Guttmann\footnote[1]{email:
                {\tt tonyg@ms.unimelb.edu.au}}}
\address{ ARC Centre of Excellence for\\
Mathematics and Statistics of Complex Systems,\\
Department of Mathematics and Statistics,\\
The University of Melbourne, Victoria 3010, Australia}
\author{              Mathew D Rogers\footnote[2]{email:
                {\tt mathewrogers@gmail.com}}}
\address{
Department of Mathematics and Statistics,\\
Universit\'e de Montr\'eal, Montr\'eal, Qu\'ebec, H3C\,3J7, Canada}

\end{center}
\setcounter{footnote}{0}

\begin{abstract}
We show that the integral
\begin{equation}\label{int}
J(t) = \frac{1}{\pi^3} \int_0^\pi \int_0^\pi \int_0^\pi dx dy dz \log(t-\cos{x}-\cos{y}-\cos{z}+\cos{x}\cos{y}\cos{z}),
\end{equation}
can be expressed in terms of ${_5F_4}$ hypergeometric functions. The integral arises in the solution by Baxter and Bazhanov of the free-energy of the $sl(n)$ Potts model, which includes the term $J(2).$ Our result immediately gives the logarithmic Mahler measure of the Laurent polynomial
$$k-\left(x+\frac{1}{x}\right )-\left(y+\frac{1}{y}\right )-\left(z+\frac{1}{z}\right )+\frac{1}{4}\left(x+\frac{1}{x}\right ) \left(y+\frac{1}{y}\right ) \left(z+\frac{1}{z}\right )$$ in terms of the same hypergeometric functions.

\end{abstract}

\section{Calculation of the integral}
There exists an extensive literature on solvable two-dimensional lattice models, based on the concept of commuting transfer matrices and the Yang-Baxter relation. Far fewer solutions are available for three-dimensional models. The first such model was obtained by Zamolodchikov \cite{Z80, Z81} in 1980-81. He introduced a so-called tetrahedron relation as the appropriate generalization of the Yang-Baxter equation. In 1983 Baxter solved this model
using only commutativity, symmetry and a factorizability
property of the transfer matrix. A decade later Bazhanov and Baxter \cite{BB93} studied the $sl(n)$-chiral Potts model \cite{B91}, which can be considered a multi-state generalization of the Zamolodchikov model, and showed that for this model it was also possible to bypass the tetrahedron equations and solve the model by using symmetry properties to prove the commutativity of the row-to-row transfer matrices. Aspects of the solution were later discussed by Bazhanov \cite{B93}, who showed that the solution could be expressed in terms of a free boson or Gaussian model on the three-dimensional cubic lattice. The partition function was expressed in terms of a fractional power of the determinant of a cyclic square matrix. In the thermodynamic limit the free energy was given as the sum of two terms, one of which was the integral $J(2)$, where
\begin{equation}\label{int}
J(t) = \frac{1}{\pi^3} \int_0^\pi \int_0^\pi \int_0^\pi dk_1 dk_2 dk_3 \log(t-c_1-c_2-c_3+c_1 c_2 c_3),
\end{equation}
where $c_i$ denotes $\cos(k_i).$ Baxter and Bazhanov showed that
\begin{equation}\label{J(2) evaluation}
J(2)=\frac{8}{\pi}L_{-4}(2) - 3\log{2},
\end{equation}
where $L_{-4}(2)$ is Catalan's constant, $L_{k}(s):=\sum_{n=1}^{\infty}\left(\frac{k}{n}\right)n^{-s}$ denotes Dirichlet's $L$-series, and $\left(\frac{k}{n}\right)$ is the Jacobi symbol.

Here we show that the integral can be expressed in terms of $_5F_4$ hypergeometric functions for $|t|$ sufficiently large, and $t\ge 2.0802\dots$ on the real axis. Our starting point is the work of Delves, Joyce and Zucker \cite{J98} who considered the related, Green function integral
\begin{equation}\label{green}
G(t) = \frac{1}{\pi^3} \int_0^\pi \int_0^\pi \int_0^\pi \frac {dk_1 dk_2 dk_3}{t-c_1-c_2-c_3+c_1 c_2 c_3}.
\end{equation}
If $t$ lies in the complex plane cut along the real axis from $-2$ to $2,$ then they proved that
\begin{equation}\label{Joyce result}
G(t) = \frac{1}{t}\left ( 1- \frac{4}{t^2} \right )^{-1/4} \left [ {_2F_1}\left (\substack {\frac{1}{8},\frac{5}{8}\\1};\frac{4}{t^2} \right ) \right ]^2.
\end{equation}
There are at least two methods for integrating this formula.  The more difficult approach is to use \eqref{Joyce result} to derive a modular expansion for $J(t)$, which can then be compared to known modular expansions for $_5F_4$ functions.  This approach was used in \cite{R09} to study three-variable Mahler measures.  We leave those calculations for a future paper \cite{RgNEW}, since the following theorem can be proved via standard hypergeometric transformations.
\begin{theorem} Suppose that $|\alpha|$ is sufficiently small but non-zero.  Then
\begin{equation}\label{J in terms of hypergeometrics}
\begin{split}
J&\left(\sqrt{4\alpha(1-\alpha)}+\frac{1}{\sqrt{4\alpha(1-\alpha)}}\right)\\
&=-\frac{1}{2}\log\left(4\alpha (1-\alpha)^{19} (1+\alpha)^{12} \right)-\frac{11}{4}\alpha(1-\alpha){_5F_4}\left(\substack{\frac{3}{2},\frac{3}{2},\frac{3}{2},1,1\\2,2,2,2};4\alpha(1-\alpha)\right)\\
&\quad-\frac{7\alpha }{4 (1-\alpha)^2}{_5F_4}\left(\substack{\frac{3}{2},\frac{3}{2},\frac{3}{2},1,1\\2,2,2,2};-\frac{4 \alpha}{(1-\alpha)^2}\right)
+\frac{9 \alpha(1-\alpha)^2 }{4(1+\alpha)^4}{_5F_4}\left(\substack{\frac{5}{4},\frac{3}{2},\frac{7}{4},1,1\\2,2,2,2};\frac{16 \alpha(1-\alpha)^2}{(1+\alpha)^4}\right).
\end{split}
\end{equation}
\end{theorem}
\begin{proof}  We show that the derivatives of both sides of the equation agree.  To simplify the calculations, briefly assume that $\alpha$ is a small positive real number.  Notice that for $s\in(0,1)$:
\begin{equation}\label{derivative formula}
\frac{d}{d z}\left[z~{_5F_4}\left(\substack{2-s,\frac{3}{2},1+s,1,1\\2,2,2,2};z\right)\right]=\frac{2}{s(1-s)z}\left[{_3F_2}\left(\substack{1-s,\frac{1}{2},s\\1,1};z\right)-1\right].
\end{equation}
Let us set
\begin{equation}\label{t-def}
t:=\sqrt{4\alpha(1-\alpha)}+\frac{1}{\sqrt{4\alpha(1-\alpha)}}.
\end{equation}
Differentiate \eqref{J in terms of hypergeometrics} with respect to $\alpha$, and then apply \eqref{green}, \eqref{Joyce result}, and \eqref{derivative formula}, to obtain the presumed equality
\begin{equation}\label{presumed equality}
\begin{split}
\frac{1}{t}\left ( 1- \frac{4}{t^2} \right )^{-1/4} &\left [ {_2F_1}\left (\substack {\frac{1}{8},\frac{5}{8}\\1};\frac{4}{t^2} \right ) \right ]^2~\frac{2(2 \alpha-1)^3}{(4\alpha(1-\alpha))^{3/2}}\\
=&-\frac{11 (1-2 \alpha) }{2 \alpha (1-\alpha)}\, _3F_2\left(\substack{\frac{1}{2},\frac{1}{2},\frac{1}{2}\\1,1};4 \alpha (1-\alpha)\right)
+\frac{7 (1+\alpha) }{2 \alpha (1-\alpha)}\, _3F_2\left(\substack{\frac{1}{2},\frac{1}{2},\frac{1}{2}\\1,1};-\frac{4 \alpha}{(1-\alpha)^2}\right)\\
&+\frac{3 \left(1-6\alpha+\alpha^2\right) }{2 \alpha \left(1-\alpha^2\right)}\, {_3F_2}\left(\substack{\frac{1}{4},\frac{1}{2},\frac{3}{4}\\1,1};\frac{16 \alpha(1-\alpha)^2}{(1+\alpha)^4}\right).
\end{split}
\end{equation}
We can re-express the left-hand side of the identity using equation (16) on p.~112 of Erd\'elyi et al. \cite{Bv1}, giving
\begin{equation*}
\left ( 1- \frac{4}{t^2} \right )^{-1/4} \left [ {_2F_1}\left (\substack {\frac{1}{8},\frac{5}{8}\\1};\frac{4}{t^2} \right ) \right ]^2 = \xi \left [ {_2F_1}\left (\substack {\frac{1}{4},\frac{3}{4}\\1};\frac{1-\xi}{2}\right ) \right ]^2,
\end{equation*}
where $\xi =\left (1-\frac{4}{t^2}\right )^{-1/2}$. Substituting for $t$ yields
\begin{align}
\left ( 1- \frac{4}{t^2} \right )^{-1/4} \left [ {_2F_1}\left (\substack {\frac{1}{8},\frac{5}{8}\\1};\frac{4}{t^2} \right ) \right ]^2 = &\frac{1+4 \alpha-4 \alpha^2}{(1-2 \alpha)^2} \left [ {_2F_1}\left (\substack {\frac{1}{4},\frac{3}{4}\\1};-\frac{4\alpha(1-\alpha)}{(1 - 2 \alpha)^2}\right ) \right ]^2\notag\\
=&\frac{1+4 \alpha-4 \alpha^2}{(1-\alpha)(1-2 \alpha)}\left[{_2F_1}\left(\substack{\frac{1}{2},\frac{1}{2}\\1},\frac{\alpha}{\alpha - 1}\right)\right]^2\notag\\
=&\frac{1+4 \alpha-4 \alpha^2}{(1-2 \alpha)}\left[{_2F_1}\left(\substack{\frac{1}{2},\frac{1}{2}\\1},\alpha\right)\right]^2,\label{trans 1}
\end{align}
where the second and third steps follow from \cite[pg.~95]{Be2} and \cite[pg.~38]{Be2}
The right-hand side of \eqref{presumed equality} simplifies via Clausen's identity, and the same quadratic transformations:
\begin{align}
{_3F_2}\left(\substack{\frac{1}{2},\frac{1}{2},\frac{1}{2}\\1,1};4 \alpha (1-\alpha)\right)=&\left[_2F_1\left(\substack{\frac{1}{2},\frac{1}{2}\\1};\alpha\right)\right]^2,\label{trans 2}\\
{_3F_2}\left(\substack{\frac{1}{2},\frac{1}{2},\frac{1}{2}\\1,1};-\frac{4 \alpha}{ (1-\alpha)^2}\right)=&\left[_2F_1\left(\substack{\frac{1}{2},\frac{1}{2}\\1};\frac{\alpha}{\alpha-1}\right)\right]^2=(1-\alpha)\left[_2F_1\left(\substack{\frac{1}{2},\frac{1}{2}\\1};\alpha\right)\right]^2\label{trans 3}\\
{_3F_2}\left(\substack{\frac{1}{4},\frac{1}{2},\frac{3}{4}\\1,1};\frac{16 \alpha(1-\alpha)^2}{(1+\alpha)^4}\right)=&\left[_2F_1\left(\substack{\frac{1}{4},\frac{3}{4}\\1};\frac{4\alpha}{(1+\alpha)^2}\right)\right]^2=(1+\alpha)\left[_2F_1\left(\substack{\frac{1}{2},\frac{1}{2}\\1};\alpha\right)\right]^2\label{trans 4}
\end{align}
If we substitute \eqref{t-def}, \eqref{trans 1}, \eqref{trans 2}, \eqref{trans 3}, and \eqref{trans 4} into \eqref{presumed equality}, we see that \eqref{presumed equality} holds whenever $\alpha$ is a sufficiently small real number.  This implies that \eqref{J in terms of hypergeometrics} holds up to a constant of integration.  The constant of integration is easily seen to equal zero, because both sides of \eqref{J in terms of hypergeometrics} approach $-\frac{1}{2}\log\alpha+0$ when $\alpha$ tends to zero.  Finally, if we add $\frac{1}{2}\log\alpha$ to either side of the identity, then both sides are analytic in a neighborhood of $\alpha=0$, so \eqref{J in terms of hypergeometrics} holds for $|\alpha|$ sufficiently small but non-zero.
\end{proof}
The formula for $J(t)$ holds on the positive real axis for $0<\alpha\le(\sqrt{2}-1)^2\approx.1715...$, which implies $t\ge 2.0802\dots$.  The identity fails when $\alpha>(\sqrt{2}-1)^2$, because the argument of the third hypergeometric function crosses a branch cut which lies on $[1,\infty)$.  In general, we can analytically continue \eqref{J in terms of hypergeometrics} along a ray starting from $\alpha=0$, and ending at a point where one of the functions ceases to be analytic.  Despite the fact that equation \eqref{J in terms of hypergeometrics} can not be used to calculate $J(2)$, it is still possible to reprove the formula of Baxter and Bazhanov via a closely related modular expansion \cite{RgNEW}:
\begin{equation}\label{J partial fractions}
\begin{split}
J\left(u\left(e^{-2\pi v}\right)\right)=&-3\log2 +\frac{15v}{\pi^3}\sum_{(n,k)\ne (0,0)}\frac{3n^2-(2v)^2 k^2}{\left(n^2+(2v)^2 k^2\right)^3}\\
&+\frac{48v}{\pi^3}\sum_{n,k}\frac{3(2n+1)^2-(2v)^2(2k+1)^2}{\left((2n+1)^2+(2v)^2(2k+1)^2\right)^3},
\end{split}
\end{equation}
where
\begin{align}\label{u-formula}
u(q)=\left(\sqrt[4]{2}\frac{\eta(q)\eta\left(q^4\right)}{\eta^2\left(q^2\right)}\right)^{12} +\left(\sqrt[4]{2}\frac{\eta(q)\eta\left(q^4\right)}{\eta^2\left(q^2\right)}\right)^{-12},
\end{align}
and $\eta(q)=q^{1/24}\prod_{n=1}^{\infty}(1-q^n)$.  Equation \eqref{J partial fractions} can be derived by combining equation \eqref{J in terms of hypergeometrics} with (2.12) and (2.14) in \cite{R09}.  It is possible to show that \eqref{J partial fractions} holds for $v\ge \frac12$, and that the left-hand side of the identity equals $J(2)$ when $v=\frac12$.  The right-hand side reduces to two-dimensional lattice sums, which can ultimately be evaluated by appealing to results of Glasser and Zucker \cite{GZ}.  Bertin and Rodriguez-Villegas have both used a similar modular approach to prove certain Mahler measure formulas \cite{Be}, \cite{RV}.

Equation \eqref{J in terms of hypergeometrics} also implies an identity between Mahler measures.  The Mahler measure of an $n$-variable Laurent polynomial, $P(x_1, \dots, x_n)$, is defined by
\begin{equation*}
\m(P):=\int_{0}^{1}\dots\int_{0}^{1}\log\left|P\left(e^{2\pi i t_1},\dots,e^{2\pi i t_n}\right)\right|d t_1\dots d t_n.
\end{equation*}
The second author related both $_5F_4$ functions to Mahler measures in \cite{R09}.  After some simplification, equation \eqref{J in terms of hypergeometrics} implies
\begin{align*}
\m&\left(8\sqrt{4\alpha(1-\alpha)}+\frac{8}{\sqrt{4\alpha(1-\alpha)}}-4\left(x+x^{-1}+y+y^{-1}+z+z^{-1}\right)+\left(x+x^{-1}\right)\left(y+y^{-1}\right)\left(z+z^{-1}\right)\right)\\
&\quad=11\m\left(\frac{4}{\sqrt{\alpha(1-\alpha)}}+\left(x+x^{-1}\right)\left(y+y^{-1}\right)\left(z+z^{-1}\right)\right)\\
&\qquad-7\m\left(\frac{4i(1-\alpha)}{\sqrt{\alpha}}+\left(x+x^{-1}\right)\left(y+y^{-1}\right)\left(z+z^{-1}\right)\right) \\
&\qquad-6\m\left(x^4+y^4+z^4+1+\frac{2(1+\alpha)}{\sqrt[4]{\alpha(1-\alpha)^2}}x y z\right),
\end{align*}
and this identity also holds on the real axis for $\alpha\in(0,(\sqrt{2}-1)^2]$.
\section{Conclusion and special values of $J(t)$}

We conclude by noting that formula \eqref{J(2) evaluation} of Baxter and Bazhanov is not an isolated result, and  that there are many additional explicit formulas for values of $J(t)$.
Most of these formulas involve $L$-functions of eta products.  These are not elementary constants, but they often carry deep number-theoretic significance, so in some sense they can still be regarded as fundamental constants.  For more details
on explicit Mahler measure formulas, we refer to the work of Boyd \cite{B98}.  We conclude with a single example of such a formula.  If we have a function $$f(\tau)=\sum_{n=1}^{\infty}a_n e^{2\pi i n \tau},$$ then the $L$-series associated with $f$ is defined by
\begin{equation*}
L(f,s):=\sum_{n=1}^{\infty}\frac{a_n}{n^s}.
\end{equation*}
Thus if we take $\eta(\tau)$ to be the usual Dedekind eta function, it is possible to prove
\begin{align}
J\left(\frac{5}{2}\right)=&\frac{24 \sqrt{3}}{\pi^3}L\left(\eta^3(2\tau)\eta^3(6\tau),3\right) + \frac{15 \sqrt{3}}{4 \pi}L_{-3}(2) - 3 \log2.\label{explicit J(5/2) value}
\end{align}
We have also obtained more complicated identities for $J(14)$, $J(322)$, and $J(t)$ for various irrational algebraic values of $t$ \cite{RgNEW}.

It is also noteworthy that the two hypergeometric functions that appear in Theorem 1 are precisely those that appear in our solution \cite{GR12} of the spanning tree constant for the simple-cubic lattice.

\section*{Acknowledgements}
AJG wishes to acknowledge helpful conversations with Vladimir Bazhanov, and financial support from the Australian Research Council through grant DP1095291.

\end{document}